%% file: arxiv.tex
\documentclass{article}

\pdfoutput=1

\usepackage[utf8]{inputenc} % allow utf-8 input
\usepackage[T1]{fontenc}    % use 8-bit T1 fonts
\usepackage{url}            % simple URL typesetting
\usepackage{booktabs}       % professional-quality tables
\usepackage{amsfonts}       % blackboard math symbols
\usepackage{nicefrac}       % compact symbols for 1/2, etc.
\usepackage{microtype}      % microtypography
\usepackage{lipsum}
\usepackage{graphicx}
\usepackage{algorithm}
\usepackage{algorithmic}
\usepackage[switch]{lineno}
\usepackage{dirtytalk}
\usepackage{mathtools}
\usepackage{amsthm}
\usepackage{authblk}
\usepackage{tikz}

\usepackage{natbib}
\usepackage[verbose=true,letterpaper]{geometry}
\AtBeginDocument{
  \newgeometry{
    textheight=9in,
    textwidth=6.5in,
    top=1in,
    headheight=14pt,
    headsep=25pt,
    footskip=30pt
  }
}

\usepackage[colorlinks, citecolor=teal, linkcolor=magenta, urlcolor=teal]{hyperref}

\setcitestyle{authoryear, open={(},close={)}}

\newtheorem{example}{Example}
\newtheorem{theorem}{Theorem}
\newtheorem{lemma}{Lemma}
\newtheorem{definition}{Definition}

\newtheorem{corollary}{Corollary}
\newtheorem{claim}{Claim}
\newtheorem{proposition}{Proposition}[section]

\newtheorem*{remark}{Remark}

\title{Fair and Truthful Allocations Under Leveled Valuations}
\date{}
\author[1,2]{
George Christodoulou
}
\author[1,2]{
Vasilis Christoforidis}
\affil[1]{Archimedes/Athena RC}
\affil[2]{Aristotle University of Thessaloniki}

\begin{document}
    \maketitle
    
    \input{abstract}
    \input{intro}
    \input{contributions}

    \input{related}
    \input{prelims}

    \input{mms}

    \input{efx}

    \input{truthfulness}

    \input{conclusion}

    \input{acknowledgments}

    \bibliography{bibliography}  
\end{document}

%% file: abstract.tex
\begin{abstract}
We study the problem of fairly allocating indivisible goods among
agents which are equipped with {\em leveled} valuation functions. Such
preferences, that have been studied before in economics and fair
division literature, capture a simple and intuitive economic
behavior; larger bundles are always preferred to smaller ones. We provide a fine-grained analysis for various subclasses of leveled valuations focusing on two extensively studied notions of fairness, (approximate) MMS and EFX. In particular, we present a general positive result, showing the existence of $2/3$-MMS allocations under valuations that are both leveled and submodular. We also show how some of our ideas can be used beyond the class of leveled valuations; for the case of two submodular (not necessarily leveled) agents we show that there always exists a $2/3$-MMS allocation, complementing a recent impossibility result. Then, we switch to the case of subadditive and fractionally subadditive leveled agents, where we are able to show tight (lower and upper) bounds of $1/2$ on the approximation factor of MMS. Moreover, we show the existence of exact EFX allocations under general
leveled valuations via a simple protocol that in addition satisfies
several natural economic properties. Finally, we take a mechanism design approach and we propose protocols that are both truthful and approximately fair under leveled valuations.
\end{abstract}

%% file: intro.tex
\section{Introduction}
\label{sec:intro}

Allocating indivisible resources in a fair 
manner has been of central importance in economics and computer science. The fundamental notions of proportionality and envy-freeness provide strong existential results for divisible goods, but their guarantees do not carry over to the indivisible setting. Therefore, we resort to various relaxations and approximations thereof, such as the (approximate) maximin share and envy-freeness up to any item.

The \emph{maximin share guarantee} (MMS) is a compelling and well-studied notion proposed by \citet{BudishMMS} which can be seen as a natural extension of the well-known cut-and-choose protocol. The goal here is to allocate a bundle to each agent, ensuring that its value is at least as high as her maximin share (or a large fraction thereof). Intuitively, the maximin share corresponds to the maximum value an agent can get after proposing an initial allocation and keeping the least desirable bundle for herself. Such allocations might not exist for additive valuations \citep{KurokawaProcacciaWangMMS}. The study of MMS allocations and approximations thereof (mostly under additive utilities) has received significant attention over the last few years \citep{KurokawaProcacciaWangMMS,AmanatidisMarkakisNikzadSaberiMMSapx,BarmanKrishnamurthyMMS,GargTakiMMSimprv}.

In sharp contrast to weaker relaxations, the existence of \emph{envy-freeness up to any good} (EFX) allocations \citep{CaragiannisEFX} is yet to be settled beyond a handful of simple settings \citep{PlautRoughgarden,ChaudhuryGargMehlhorn}.
A reasonable approach is to focus on restricted valuation spaces and approximations.

In this work, we focus on a specific class of valuations, namely \emph{leveled} valuations.
They capture a simple economic behavior in which any set of items is strictly preferred to any other set of smaller cardinality, regardless of the specific items included; among bundles of equal size any preference ordering is feasible.\footnote{Leveled preferences can be perceived as a way of prioritizing quantity over quality, thus giving rise to the definition of \emph{quantity-monotonic} preferences used by \citet{PapaiQuotas}.} \citet{BNTCCompetitiveJournal2021} study leveled preferences in the context of competitive equilibrium in markets, while \citet{GafniCopies} study EFX allocations under leveled preferences for chores. We also note that leveled preferences might arise in the context of web advertising where some competitors aim to increase their total coverage rather than focusing on specific advertising spots in the ad space allocation. Lastly, leveled preferences might appear under budget constraints: when resources are scarce, agents facing budget constraints might prioritize quantity to ensure that they have enough items to meet their needs.

%% file: contributions.tex
\subsection{Our Contributions}
\label{subsec:contributions}

We study fair allocation in a setting where a set of $m$ indivisible
goods need to be allocated to $n$ agents in a {\em fair} manner. We
focus on the case where the agents are equipped with {\em leveled}
valuations and study two well-studied notions of fairness, (approximate)
MMS and EFX.

 \paragraph{(approximate) MMS.} First, in Section~\ref{sec:mms} we focus on (approximate) MMS and we are
able to show that every instance where $m<2n$ admits an exact MMS
solution. In stark contrast, when $m\geq 2n$ we observe that there are
instances (even for two agents and four items) where there is no
approximate MMS solution, for any approximation factor. 

In light of this negative result, and in order to provide a more
fine-grained analysis of approximate MMS under leveled valuations, we
focus on important valuation subclasses such as submodular leveled
(Section \ref{sec:submodMMS}), fractionally subadditive (XOS) and subadditive leveled functions
(Section \ref{sec:subadditive-mms}). Submodular, XOS, and subadditive valuations capture complex combinatorial preferences and have
been studied extensively in microeconomics and algorithmic game
theory \citep{LehmannLehmannNisan,Feige09,BayesianCombAuctions,FeldmanFGL20}.

In Section \ref{sec:submodMMS}, in our main technical result, we
show the existence of $2/3$-MMS allocations for any number of agents
and items when agents are equipped with submodular leveled
valuations. We note that our positive result (lower bound) complements
the recent impossibility result (upper bound) of
\citet{KulkarniKulkarniMehtaMMS} that states that there are instances
with two agents where a guarantee better than $2/3$ is impossible to
achieve. We demonstrate how some of our proof ideas may help to provide
better bounds for valuations beyond leveled. Indeed, in Theorem
\ref{thm:2submod} we are able to show that for the case of two
submodular (not necessarily leveled) agents and multiple goods there
is always a $2/3$-MMS allocation. We highlight the fact that the nice structure of leveled valuations does not yield any stronger bounds in this case. Next, we switch to fractionally subadditive (XOS) and subadditive leveled valuations where we are able to
provide tight approximate MMS bounds.
\paragraph{EFX.} In Section \ref{section:efx} we shift our attention to EFX
allocations. We prove the existence of exact EFX allocations under
leveled valuations via a simple protocol. Moreover, our algorithm works even when only ordinal information is given as input. In
particular, under strict leveled preferences our mechanism achieves a
set of natural properties in addition to EFX, 
overcoming the barriers of existing impossibility results on the
interplay between fairness, efficiency, and incentive
compatibility. Our results complement the previous work of 
\citet{GafniCopies} which showed the existence of EFX allocations under
the special case of \emph{additive} leveled valuations.\footnote{We note that the model studied in \citep{GafniCopies} is different than ours. They consider fair allocations for goods with copies and develop a duality framework between chores and goods with copies. Without copies, their proposed fairness concept coincides with the standard notion of EFX.}

\paragraph{Truthfulness.} Finally, we explore the interplay of fairness with truthfulness for the class of leveled valuations. We are able to show that constant factor MMS approximations can be achieved via truthful mechanisms under leveled valuations. This is in contrast to (non-leveled) valuations where non-constant upper bounds persist even for the case of additive valuations~\citep{AmanatidisBirmpasChristodoulouMarkakis}. Additionally, we show an algorithmic characterization of mechanisms satisfying truthfulness, Pareto optimality, non-bossiness, neutrality, and EFX.

We note that our focus is explicitly on the existence of such notions, rather than on their computational aspects.

%% file: related.tex
\subsection{Related Work}
\label{subsec:related}

In this section we discuss prior works closely related to EFX, MMS, and leveled valuations. The growing literature on fair allocation of indivisible items is too extensive to cover on this paper and thus, we refer the interested reader to the recent survey of \citet{AmanatidisSurvey} for a detailed overview.

\textbf{EFX.} The concept of EFX poses several challenges leading to important open problems in the area.  \citet{PlautRoughgarden} showed that when agents exhibit identical valuations or share the same ordinal ranking of goods, EFX is guaranteed to exist. In a breakthrough result, \citet{ChaudhuryGargMehlhorn} showed that complete, exact EFX allocations always exist for three agents under additive utilities, later improved to more general valuations  \citep{Berger_Cohen_Feldman_Fiat_2022,AkramiEC23EFX}. Other works considered a limited number of items or agents sharing one of two possible valuation types \citep{AmanatidisMarkakisNtokosBirds,MaharaTypes,MaharaExtension}, while a major line of work focused on restrictions on the valuation space, prompting a series of works on dichotomous valuations \citep{AmanatidisEFXStories,BabaioffEzraFeigeDichotomous,BenabbouMRFs,HalpernBinary}, lexicographic preferences \citep{HosseiniLexicographic}, graph instances \citep{ChristodoulouFiatKoutsoupiasSgouritsa}, and allocations of multisets \citep{GorantlaMultiset}. Recent work has shown that EFX allocations might not exist under general monotone valuations when the items to be allocated are chores \citep{CS24}. 

\textbf{MMS.} The existence of maximin fair allocations is not guaranteed for additive utilities \citep{KurokawaProcacciaWangMMS}. 
A series of works established strong approximation guarantees reaching a factor of $3/4 + \mathcal{O}(1/n)$ \citep{KurokawaProcacciaWangMMS,AmanatidisMarkakisNikzadSaberiMMSapx,BarmanKrishnamurthyMMS,GhodsiHSSYmor,GargTakiMMSimprv,akrami2023simplification}. Very recently, \citet{akrami2023breaking} broke the barrier of $3/4$. Considerable effort has been shown to the case of four or fewer agents \citep{AmanatidisMarkakisNikzadSaberiMMSapx,GourvesMonnotMMSapx,GhodsiHSSYmor,feige2022improvedMMSforThree} or settings with few items \citep{FeigeInapprox,hummel2023lower}. The current impossibility result for additive valuations stands at $1-1/n^4$ due to \citet{FeigeInapprox}. Moreover, the study of maximin fair allocations beyond the additive domain has been of growing interest. \citet{BarmanKrishnamurthyMMS, FeigeSubmodularMMS} and \citet{GhodsiArtIntBeyondAdditive, Seddighin_Seddighin_2022} provided a series of results for complement-free valuations. We briefly summarize the known results in Table \ref{table:summaryResults}.

\textbf{Leveled Preferences.}
Typical examples of such preferences mentioned in the literature include the distribution of cabinet seats among parties forming a coalition government or the allocation of offices among different departments at universities. \citet{PapaiQuotas} characterized the set of strategyproof mechanisms under standard economic assumptions in assignment problems. \citet{GafniCopies} prove the existence of EFX allocations under \emph{leveled additive} valuations, while \citet{BabaioffEzraFeigeDichotomous} consider $\epsilon$-leveled valuations (additive valuation functions that are $\epsilon$-dichotomous).  \citet{OhProcacciaSuksompong} show that there exists a deterministic algorithm that computes an EF1 allocation using $\mathcal{O}(n^2)$ queries under leveled valuations. Despite their simplicity, leveled valuations have played an important role in constructing counterexamples; we note that the first negative results for both notions (non-existence of MMS allocations and exponential query complexity for EFX) actually employ leveled valuations \citep{KurokawaProcacciaWangMMS, PlautRoughgarden}. The best currently known MMS bounds for more general classes are also (essentially) leveled.

%% file: prelims.tex
\section{Model and Preliminaries}

In this section we introduce the main concepts and notation.

\textbf{Model.} We want to allocate $m$ goods to $n$ agents. We denote as $M = \{1,...,m\}$ the set of goods and with $N = \{1,...,n\}$ the set of agents. Each agent $i$ is equipped with a valuation function $v_i: 2^M \rightarrow \mathbb{R}_{\geq 0}$.

\textbf{Allocations and bundles.} An allocation is a tuple of subsets of $M$, namely an $n$-partition $A = (A_1,...,A_n)$ where $A_i$ is the bundle allocated to agent $i$. The elements of $A$ are pairwise disjoint, that is $A_i \cap A_j = \emptyset$ for every pair of agents $i,j \in N$. $\Pi_n(M)$ denotes the set of all possible $n$-partitions of $M$. Moreover, we only focus on complete allocations in which no items remain unallocated, that is $\bigcup_{i \in N} A_i = M$. 
We denote $[k] = \{1,...,k\}$ for any $k \in \mathbb{N}$.

\textbf{Classes of valuation functions.} Throughout this paper, we consider the valuation functions to be normalized, i.e. $v(\{\emptyset\}) = 0$, and monotone ($v(S) \leq v(T)$ for $S \subseteq T$). We focus on leveled valuations.

\begin{definition}
A valuation function $v$ is leveled if for any two bundles $S, T$ with $\lvert S \rvert > \lvert T \rvert$ it holds that $v(S) > v(T)$.
\end{definition}

Moreover, we consider subclasses of complement-free valuations.

\begin{itemize}
    \item A valuation function $v$ is \emph{additive} if for every $S \subseteq M$, we have that $v(S)=\sum_{j \in S} v(\{g\})$. For ease of convenience, we may write $v(g)$ instead of $v(\{g\})$.
    
    \item A valuation function $v$ is \emph{submodular} if $v(S) + v(T) \geq v(S \cup T) + v(S \cap T)$ for all sets $S$ and $T$ in $M$.

    \item A valuation function $v$ is \emph{fractionally subadditive} (\emph{XOS}) if there exists a family of additive functions $a_1, \dots, a_k$ such that $v(S) = \max_{l \in [k]} a_l(S)$

    \item A valuation function $v$ is \emph{subadditive} if $v(S) + v(T) \geq (S \cup T)$ for all sets $S$ and $T$ in $M$.
\end{itemize}

The hierarchy of the aforementioned classes is as follows: Additive $\subset$ Submodular $\subset$ XOS $\subset$ Subadditive.

\textbf{Fairness.} In envy-free allocations, no agent would strongly prefer to swap her allocated bundle with someone else. Formally, $v_i(A_i) \geq v_i(A_j), \forall i, j \in N$. Envy-free solutions may fail to exist in settings with indivisible goods. Therefore, several relaxed versions have been introduced in the literature. We focus on EFX, which requires that envy vanishes after removing any item from the bundle of the envied agent.

\begin{definition} [Envy-freeness up to any good (EFX)] 
An allocation $A \in \Pi_n(M)$ is called envy-free up to any good (EFX) if $\forall i,j \in N, \forall g \in A_j: v_i(A_i) \geq v_i(A_j \setminus \{g\})$.
\end{definition}

Another important notion is that of maximin share. $\mu_i^n(M)$ denotes the maximin share for an agent $i$, that is the maximum value she can obtain after proposing an $n$-partition of the items and securing the worst bundle for herself.  

\begin{definition} [Maximin share]
An allocation $A$ is said to be maximin share fair (MMS) if $v_i(A_i) \geq \mu_i^n(M) = \max_{A \in \Pi_n(M)} \min_{k \in [n]} v_i(A_k), \forall i \in N$.
\end{definition}

 When $n$ and $M$ are clear from context, we may write $\mu_i$ for notation simplicity. Furthermore, without loss of generality, we assume that $\mu_i^n(M) = 1$ for every agent $i \in N$. Since maximin share allocations do not always exist, we focus on multiplicative approximations of the maximin share. 

\begin{definition} [$\alpha$-MMS]
Let $\alpha \in (0,1]$. An allocation $A$ is said to be $\alpha$-maximin fair if it guarantees to an agent an $\alpha$-fraction of her MMS, that is $v_i(A_i) \geq \alpha \cdot \mu_i^n(M), \forall i \in N$.
\end{definition}

\textbf{Mechanisms and Truthfulness.} A deterministic allocation mechanism without payments, henceforth a \emph{mechanism}, is a mapping $\mathcal{X}$ from valuation profiles to allocations. We let $\mathcal{X}_i(\boldsymbol{v})$ denote the set of items agent $i$ receives. A mechanism $\mathcal{X}$ is truthful if for any instance $\boldsymbol{v} = (v_1, \dots, v_n)$, any player $i \in N$, and any $v_i'$: $v_i(X_i(\boldsymbol{v})) \geq v_i(X_i(v_i', v_{-i}))$

%% file: mms.tex
\section{Maximin share (MMS)}
\label{sec:mms}

In this section, we focus on approximate MMS and show general positive and negative results under leveled valuations. First, we focus on instances with few items; we show that every instance with $m<2n$ admits an exact MMS allocation under general leveled valuations.\footnote{Similar settings with a small number of items have been explored in a surge of works for additive utilities \citep{BouveretLemaitreMMS,KurokawaProcacciaWangMMS,FeigeInapprox,hummel2023lower}.} Exceeding this threshold may lead to instances where the approximation ratio becomes unbounded. 
Subsequently, we proceed to a more refined analysis of the concept for valuations that are both leveled and submodular. We show a general positive result, establishing the existence of $2/3$-MMS allocations under submodular leveled valuations (Theorem \ref{thm:existence-submodular}). We then switch to subadditive and fractionally subadditive agents and show tight lower and upper bounds on the approximation of the maximin share (Section \ref{sec:subadditive-mms} and Section \ref{sec:upperBounds}).

We now state our proposition concerning scenarios where the number of items is limited in relation to the number of agents, specifically when $m < 2n$.

\begin{proposition} \label{prop:fewItems}
    Every instance where $m<2n$ admits an exact MMS allocation under general leveled valuations.
\end{proposition}

\begin{proof}
    For $n \geq m$ verifying the existence of MMS allocations is straightforward. Let's consider the case of $m=n+r$ for some positive $r<n$. Note that in that case, due to leveled valuations, the MMS partition of each agent consists of $n-r$ singletons and $r$ pairs, and the MMS value corresponds to the value of the $n-r$-th best singleton.
    Let $\sigma$ be some arbitrary permutation of $[n]$. We let the first $n - r$ agents w.r.t $\sigma$ to choose their preferred item. Subsequently, we allocate the rest of the items in pairs to the remaining agents in an arbitrary manner. Clearly, agents receiving two items guarantee their full MMS due to leveled valuations. In the case of an agent $i$ receiving a single item, it holds that there are are at most $n-r-1$ agents that precede her in $\sigma$ each getting at most one item. Therefore $i$ can still guarantee her MMS value.
\end{proof}

Adding more items beyond this point transitions us from exact MMS solutions to instances with unbounded approximation ratio, even for two agents and four goods.

\begin{example}
\normalfont
Consider an instance with $n$ agents and $m=2n$ goods $\{g_1, \dots, g_m\}$. For all agents $i \in [n-1]$ we have that $v_i(S) = \epsilon$ for all $S$, such that $\lvert S \rvert = 1$. Furthermore,  $v_i(\{g_{2j-1}, g_{2j}\}) = 1, \forall j \in [n], \forall i \in [n-1]$, while any other pair of items is valued at $\epsilon + \delta$, with $\epsilon, \delta$ small positive constants. For agent $n$ we have $v_n(S) = \epsilon$ for all $S$, $\lvert S \rvert = 1$, and $v_n(S') = 1$ for all bundles $S' \in \mathcal{S} = (\{g_1,g_{2n}\}, \{g_2, g_3\}, \dots, \{g_{2n-2}, g_{2n-1}\})$; the rest of the pairs are also valued at $\epsilon + \delta$. Lastly,  $v_i(S) = \lvert S \rvert + \epsilon$ for all $\lvert S \rvert > 2$. One can easily check that the valuations described above are leveled. The proof follows the structure of the proof of Corollary \ref{cor:upperSubadditive} due to \citet{GhodsiArtIntBeyondAdditive}, where all but one agent share the same valuation function. The maximin share of each agent is equal to $1$. We observe that one must receive a low-valued bundle. It is easy to see that we can make $\epsilon$ infinitesimally small while keeping the valuations leveled, thus making any approximation guarantee arbitrarily bad.

Hence, we shift our attention to subclasses of leveled valuation functions, such as submodular, fractionally subadditive, and subadditive leveled valuations. Our results are summarized in Table \ref{table:summaryResults}. 
\end{example}

\begin{table}[]
    \centering
    \begin{tabular}{c|c|c}
         \textbf{Submodular} & Existence & Non-Existence  \\
         \hline
         $n = 2$ & $\boldsymbol{2/3}$ & $2/3$ \\
         $n \geq 3$ & $10/27$ & $3/4$\\
         \hline
         \textbf{Submodular Leveled} \\
         \hline
         $n = 2$ & $\boldsymbol{2/3}$ & $2/3^*$ \\
         $ n \geq 3$ & $\boldsymbol{2/3}$ & ${3/4}^*$ \\

         \hline
         \textbf{XOS} \\
         \hline
         $n \geq 2$ & $3/13$ & $1/2$ \\
         \hline
         \textbf{XOS Leveled} \\
         \hline
         $n \geq 2$ & $\boldsymbol{1/2}$ & $1/2^*$ \\

         \hline
         \textbf{Subadditive}\\
         \hline
         $n \geq 2$ & $\Omega(1/ \log{n} \log\log{n})$ & $1/2$ \\
         \hline
         \textbf{Subadditive Leveled} \\
         \hline
         $n \geq 2$ & $\boldsymbol{1/2}$ & $1/2^*$
         
    \end{tabular}
    \caption{Best known approximations of MMS with submodular agents are due to \citet{GhodsiArtIntBeyondAdditive, FeigeSubmodularMMS, KulkarniKulkarniMehtaMMS, Seddighin_Seddighin_2022}. The results on XOS valuations are due to \citet{GhodsiArtIntBeyondAdditive} and \citet{akrami2023randomizedXOS}, while bounds on subadditive valuations come from \citep{GhodsiArtIntBeyondAdditive} and \citep{Seddighin_Seddighin_2022}. Our results appear in \textbf{bold}. The upper bounds noted with $(^*)$ are not explicitly stated in previous works, but one can tweak known counterexamples to also work for leveled valuations.}
    \label{table:summaryResults}
\end{table}

\subsection{Submodular Valuations} \label{sec:submodMMS}
In this section we provide a thorough analysis w.r.t approximate MMS
for the case of valuations that are both leveled and
submodular. First, we present a general positive result
(Theorem~\ref{thm:existence-submodular}), showing that a $2/3$-MMS
allocation always exists, and we present a procedure that computes 
one. We argue that our proof ideas can help to provide better bounds
for valuations beyond leveled. Indeed in Theorem \ref{thm:2submod} we are able to show
that for the case of two submodular (not necessarily leveled) agents
and multiple goods there always exists a $2/3$-MMS allocation. We present our main technical result which is the
existence of a $2/3$-MMS allocation for submodular leveled
valuations. 
\begin{theorem}\label{thm:existence-submodular}
Every instance where agents have submodular leveled valuations admits a $2/3$-MMS allocation.
\end{theorem}

\begin{proof}
  We split the proof into two different cases\footnote{There is also a
    third case, where $m<2n$, but this is treated separately by
    Proposition \ref{prop:fewItems} which holds for general leveled valuations and provides
    exact MMS.}  depending on the range of $m$, namely for
  $2n\leq m <3n$ and $m\geq 3n$.  In both cases, the main element is a
  simple combinatorial argument that guarantees the abundance of
  bundles of size $\lfloor m/n \rfloor$ with values higher than $2/3$
  (Claim~\ref{claim:reshuffle} and Claim~\ref{claim:easyVersion}
  respectively). At a high level we present a partition of the items
  which is MMS for one of the agents, and we offer its bundles to the
  rest of the agents. The two claims secure that either this partition
  achieves the desired MMS guarantee of $2/3$ for all agents (case 2) or
  that there is some reallocation of the goods that achieves this (case 1).

\textbf{Case 1: $2n\leq m <3n$.} Notice that in this case, which is the more challenging one, the MMS
value for each agent is achieved by a subset of items of size equal to
$2$. 

\begin{claim} \label{claim:reshuffle} Fix an agent $i\in N$, and an
  arbitrary subset  $M'=\{a,b,c,d\}$ of four items of $M$. Let
  $S=\{a,b\}, T=\{c,d\}$ be an arbitrary partition of $M'$ into two equal
  sized sets. Let $\mathcal{U}$ be the family of all pairs of items that
  intersect both $S$ and $T$, i.e.  $ {\mathcal U}=\left\{\{a,c\},\{a,d\},\{b,c\},\{b,d\}\right\}$. If $\max\{v_i(S),v_i(T)\}<2/3$ then $v_i(U)> 2/3$ for every $U\in {\mathcal U}$. 
\end{claim}

\begin{proof}
    Take a $U\in {\mathcal U}$, without loss of generality let this be $U=\{a,c\}$. Recall that the MMS of agent $i$ is achieved by a set of size equal to $2$. We obtain 
    \begin{align*}
        v_i(S) + v_i(U) \geq v_i(S\cup \{c\})+v_i(\{a\})> 1+ v_i(\{a\}), 
    \end{align*}
    where the first inequality holds due to submodularity and the
    second due to the fact that $|S\cup \{c\}|=3$ and hence (by
    definition of leveled valuations) the value obtained by this set must be strictly greater than her MMS value. Similarly, we obtain
    \begin{align*}
        v_i(T) + v_i(U) \geq v_i(T\cup \{a\})+v_i(\{c\})> 1+ v_i(\{c\}). 
    \end{align*}
By summing these up, and by using again submodularity we obtain 
    \begin{equation*}
        v_i(S) + v_i(T) + v_i(U) > 2.
    \end{equation*}
Since $\max\{v_i(S),v_i(T)\}<2/3$ we conclude that $v_i(U)>2/3$ as needed.
\end{proof}

We now proceed with the proof of Case 1, thus showing that every instance where agents have submodular leveled valuations, $2n\leq m<3n$ admits a $2/3$-MMS allocation. Our goal is to provide a $2/3$-MMS allocation $(A^*_1,\ldots, A^*_n)$,
  such that $v_i(A^*_i)\geq 2/3$ for every agent $i$.  Initially we
  consider a partition $A=(A_1,\ldots, A_n)$ which is $2/3$-MMS
  allocation from the perspective of agent $1$, i.e.,
  $v_1(A_j)\geq 2/3$, for every $j=1,\ldots, n$. Notice that due to
  the fact that the valuations are leveled, it must be that $2\leq |A_j|\leq 3$ for all $j$. We
  offer these bundles to agents $2,\ldots, n$; as long as there exists
  an agent $i>1$ and a bundle $A_j$ such that $v_i(A_j) \geq 2/3$,
  then we allocate $A_j$ to $i$, i.e. $A^*_i:=A_j$ and we proceed with
  the remaining agents and bundles\footnote{If there are more than one
    such agents, we break ties lexicographically.}.

  If at the end of this stage each of the agents $2,\ldots, n$ is
  allocated an $A_j$ set then we allocate the remaining bundle to
  agent $1$ and the resulting allocation is $2/3$-MMS.

  Otherwise, there will be a subset of $k\leq n-1$ agents
  $N'\subseteq \{2,\ldots, n\}$ which value each of the remaining
  $k+1$ bundles
  $A_{i_1}=\{g_1,g_2\},\ldots,
  A_{i_k}=\{g_{2k-1},g_{2k}\},A_{i_{k+1}}=\{g_{2k+1},g_{2k+2}\}$ with
  value less than $2/3$. Notice that each of these bundles must have
  cardinality equal to 2, as all bundles of three items are valued
  higher than the MMS share by all agents. Let
  $M'=A_{i_1} \cup \ldots \cup A_{i_k} \cup A_{i_{k+1}}$ be the subset
  of the remaining items (it could be $M'=M$).

  But then, we can utilize Claim~\ref{claim:reshuffle} to reshuffle
  the items; If $k\geq 2$, then we can assign $A_{i_{k+1}}$ to agent $1$
  and partition the rest into $k$ pairs as follows:
  $B_1=\{g_2,g_3\},B_2=\{g_4,g_5\}, \ldots, B_k=\{g_{2k},g_1\}$. By
  Claim~\ref{claim:reshuffle} the value of each agent in $N'$ for
  every $B_j$ bundle is strictly higher than $2/3$.  If $k=1$, then
  $M'$ has four items, and again by inspection we can verify that either $B_1=\{g_2,g_3\},B_2=\{g_4,g_1\}$ or $B'_1=\{g_2,g_4\},B'_2=\{g_3,g_1\}$ will satisfy the conditions for both agents.  

\textbf{Case 2: $m \geq 3n$} The argument can be further simplified when dealing with instances of $m \geq 3n$.

\begin{claim} \label{claim:easyVersion}
    Fix an agent $i$ and two disjoint subsets $S,T$ of $M$ with $|S| =|T|=\lfloor m/n \rfloor$. If $v_i(S) < 2/3$, then $v_i(T) > 2/3$.
\end{claim}

\begin{proof}
    Take any pair of items $\{s,t\}$ with $s\in S$ and $t\in T$. By arguments similar to those used in the proof of Claim~\ref{claim:reshuffle} we obtain 
    \begin{align*}
        v_i(S) + v_i(\{s,t\}) \geq v_i(S\cup \{t\})+v_i(\{s\})> 1+ v_i(\{s\}), 
    \end{align*}
and similarly 
    \begin{align*}
        v_i(T) + v_i(\{s,t\}) \geq v_i(T\cup \{s\})+v_i(\{t\})> 1+ v_i(\{t\}). 
    \end{align*} By summing up those inequalities we get 
    \begin{align*}
    \nonumber
        v_i(S) + v_i(T) + v_i(\{s,t\}) > 2.
    \end{align*}

Finally by observing that $v_i(\{s, t\}) \leq \frac{1}{2}(v_i(S) + v_i(T))$ since the valuations are leveled we conclude that $v_i(S) + v_i(T)  > 4/3$ and the claim follows.
\end{proof}

We treat the second case in a similar but simpler fashion. We show that every instance where all agents have submodular leveled valuations, $3n\leq m$ admits a $2/3$-MMS allocation. As in case 1, we consider a partition $A=(A_1,\ldots, A_n)$
  which is an MMS allocation from the perspective of agent $1$ and
  we offer these bundles to agents $2,\ldots, n$.
  Claim~\ref{claim:easyVersion} guarantees that for each agent there
  is at most one bundle with value less than $2/3$, hence at least $n-1$ bundles with value higher than $2/3$. We let them choose sequentially in a greedy fashion, and we allocate the remaining bundle to agent 1 who values it at least as high as $2/3$. 
  \end{proof}

\subsubsection{Lower Bound for Two Submodular Agents} We complement the negative example of \citet{KulkarniKulkarniMehtaMMS} with a positive result. 

\begin{theorem} \label{thm:2submod}
The problem of fairly allocating goods to two submodular agents admits a $2/3$-MMS allocation.
\end{theorem}

\begin{proof}
    Let us consider agents $1$ and $2$ and their respective MMS partitions denoted as $P_1 = (S, \bar S)$ and $P_2 = (T, \bar T)$. We define the following sets: $A = S \cap T$, $B = S \setminus T$, $C = \bar S \cap \bar T$, $D = T \setminus S$, as depicted below: \\

    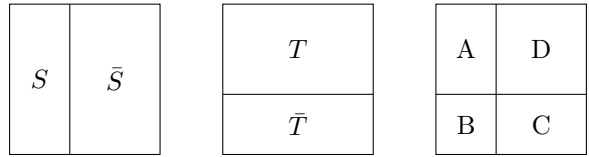
\begin{figure}[!h]
    \centering

    \begin{tikzpicture}
      \draw (0,0) -- (2,0) -- (2,2) -- (0,2) -- cycle; 
      \draw (0.8,0) -- (0.8,2); 
      \node at (0.4,1) {$S$}; 
      \node at (1.4,1) {$\bar S$}; 
    \end{tikzpicture}
    \qquad
    \begin{tikzpicture}
      \draw (0,0) -- (2,0) -- (2,2) -- (0,2) -- cycle; 
      \draw (0,0.8) -- (2,0.8); 
      \node at (1,1.4) {$T$}; 
      \node at (1,0.4) {$\bar T$};
    \end{tikzpicture}
    \qquad
    \begin{tikzpicture}
      \draw (0,0) -- (2,0) -- (2,2) -- (0,2) -- cycle; 
      \draw (0.8,0) -- (0.8,2); 
      \draw (0,0.8) -- (2,0.8);  
      \node at (0.4,1.4) {A}; 
      \node at (0.4,0.4) {B}; 
      \node at (1.4,1.4) {D};
      \node at (1.4,0.4) {C}; 
    \end{tikzpicture} \\
    \caption{A graphic interpretation of the proof given below. The first two pictures depict the MMS partitions of agent $1$ and agent $2$ respectively.}
\end{figure}

 where the vertical line represents agent $1$'s MMS partition, while the horizontal line represents agent $2$'s respective MMS partition. Therefore, the agents' MMS allocations can be written as $P_1 = (S, \bar S) = (A \cup B, C \cup D)$ and $P_2 = (T, \bar T) = (A \cup D, B \cup C)$.
 
 Following an argument similar to those used in the proof of Claim \ref{claim:reshuffle}, we have the following submodularity inequalities:
\begin{align*}
    & v_1(T) + v_1(A \cup C) \geq v_1(\bar S \cup A) + v_1(A) \\
    & v_1(\bar T) + v_1(A \cup C) \geq v_1(S \cup C) + v_1(C)
\end{align*}
Summing up the two inequalities and using submodularity again, we obtain:
\begin{align*}
    & v_1(T) + v_1(\bar T) + v_1(A \cup C) \geq 2
\end{align*}
Note that $v_1(A) + v_1(C) \geq v_1(A \cup C) + v_1(A \cap C)$ but $A \cap C = \emptyset$ by construction. The last inequality holds due to the fact that $S$ and $\bar S$ constitute the proposed MMS allocation of agent 1, and thus $v_1(S \cup C) \geq 1$ and $v_1(\bar S \cup A) \geq 1$. Similarly, for agent 2:
\begin{align*}
    v_2(S) + v_2(\bar S) + v_2(B \cup D) \geq 2
\end{align*}
We offer the bundles $T$ and $\bar T$ to agent 1 and the bundles $S$ and $\bar S$ to agent 2 respectively. Therefore, we arrive at the following scenario: either an agent accepts the offer, considering one of the proposed bundles to be valuable at least $2/3$, or both agents reject the proposed allocations. In the first case, one of the agents secures her full MMS while the other one is contented with at least $2/3$ of her MMS. In the latter situation, we identify an additional allocation $P = (\{A \cup C, B \cup D\})$ that is deemed acceptable by both agents, meaning they both value their allocated bundles at least as high as a $2/3$ fraction of their maximin shares due to the aforementioned inequalities. This completes the proof. 
\end{proof}

\subsection{Subadditive Leveled Valuations}\label{sec:subadditive-mms}
We follow a similar approach to demonstrate the existence of $1/2$-MMS allocations under subadditive valuations. Moreover, we provide a general tight upper bound by tweaking a known impossibility result due to \citet{GhodsiArtIntBeyondAdditive}. 

\begin{lemma} \label{lemma:subadditive}
    Fix an agent $i$ and two disjoint subsets $S,T$ of $M$  with $|S| =|T|=\lfloor m/n \rfloor$ and let $m \ge n$. If $v_i(S) < 1/2$, then $v_i(T) > 1/2$.
\end{lemma}

\begin{proof}
    By subadditivity, we have $      v_i(S) + v_i(T) \geq v_i(S \cup T) > 1 \nonumber$. Recall that the MMS of an agent $i$ is attained by a bundle of size $\lfloor m/n \rfloor$, thus $S \cup T$ has strictly more items than any MMS bundle. Therefore, $\max\{v_i(S),v_i(T)\}>1/2$.
\end{proof}

\begin{theorem}
    The problem of fairly allocating goods under subadditive leveled valuation functions admits a 1/2-MMS allocation. 
\end{theorem}

\begin{proof}
    Fix an agent $i$ and let $A=(A_1,\ldots, A_n)$ be her MMS allocation. We offer each agent two arbitrary bundles among the ones still available in $A$. According to Lemma \ref{lemma:subadditive}, one of the bundles is valued at least $1/2$. We let the agents pick bundles from $A$ sequentially in a greedy fashion; we allocate the remaining bundle to agent $i$.
\end{proof}

\subsection{Upper Bounds} \label{sec:upperBounds}

In this section, we provide impossibility results for approximating the MMS value for submodular (Corollary~\ref{MMS:impossiblity2sub}), XOS (Corollary~\ref{cor:upperXOS}) and subadditive (Corollary~\ref{cor:upperSubadditive}) leveled valuations which are simple adaptations of known impossibility constructions. In fact, the only necessary modification in their constructions, is to adjust the valuations to be strictly increasing in the number of items in order to make them leveled valuations. 

 \citet{KulkarniKulkarniMehtaMMS}  showcased a construction that gives an impossibility result on the existence of better than $2/3$-MMS allocations for two submodular agents. In fact, they show that the result holds even for a smaller subclass of submodular valuations, namely OXS functions. One can easily adapt their construction to make it work for leveled valuations (Corollary \ref{MMS:impossiblity2sub}).
 
 For XOS functions, we adapt the construction given in \citep{akrami2023randomizedXOS} (Table \ref{table:impossXOS}). For subsets $S$ with $\lvert S \rvert > 2$, we have $v_i(S) = 2 + \epsilon$ for $\lvert S \rvert = 3$ and $v_i(S) = 2 + 2\epsilon$ for $\lvert S \rvert = 4$. We denote the set of additive valuations of each agent as $a_1$ and $a_2$ for agent $1$, and $a_3, a_4$ for agent $2$ respectively. The XOS function receives the pointwise maximum value for each subset, i.e. $v_1(b_1, b_2) = \max \{a_1(b_1,b_2), a_2(b_1,b_2)\} = 2$. It is easy to see that the maximin value for each agent is equal to $2$. However, there is no way to guarantee MMS for both agents; one of the agents has to receive a total value of $1+\epsilon$ for an arbitrarily small $\epsilon > 0$. Finally, in order to make the counterexample general (for $n$ agents), we follow the same rationale as in Corollary \ref{cor:upperSubadditive}, adding multiple copies of \say{type 1} agents. The lower bound for XOS valuations follows from Section \ref{sec:subadditive-mms}.

\begin{table}[h]
\begin{center}
\begin{tabular}{c | c c c c }
      & $b_1$ & $b_2$ & $b_3$ & $b_4$ \\
     \hline
     $a_1$ & 1 & 1 & $\epsilon$ & $\epsilon$ \\ 
     $a_2$ & $\epsilon$ & $\epsilon$ & 1 & 1 \\
     \hline
     $a_3$ & 1 & $\epsilon$ & $\epsilon$ & 1 \\
     $a_4$ & $\epsilon$ & 1 & 1 & $\epsilon$ \\
\end{tabular}
\end{center}
\caption{Upper bound for XOS functions}
\label{table:impossXOS}
\end{table}
For subadditive valuations, one can adapt the construction given in \citep{GhodsiArtIntBeyondAdditive}.

\begin{corollary} \label{MMS:impossiblity2sub}
There exists an instance of the fair allocation problem with two submodular leveled agents in which no allocation is strictly better than 2/3-MMS \citep{KulkarniKulkarniMehtaMMS}.
\end{corollary}

\begin{corollary}\label{cor:upperXOS}
    For any $n \ge 2$ there exists an instance of the fair allocation problem with $n$ XOS leveled agents where no allocation yields more than $1/2$-MMS to the agents \citep{GhodsiArtIntBeyondAdditive}.
\end{corollary}

\begin{corollary} \label{cor:upperSubadditive}
    For any $n \geq 2$ there exists an instance of the fair allocation problem with $n$ subadditive leveled agents where no allocation yields more than $1/2$-MMS to the agents \citep{GhodsiArtIntBeyondAdditive}.
\end{corollary}

%% file: efx.tex
\section{Envy-freeness up to Any Good}
\label{section:efx}

In this section we consider envy-freeness up to any good. We establish the existence of EFX allocations under general leveled valuations. 

\textbf{Description of the algorithm}. Let $m=kn+r$ be the total number of items with $r$ being a nonnegative integer smaller than $n$ and let $k=\lfloor m/n \rfloor$. The algorithm runs sequentially. We begin by fixing a quota system in which the first $n-r$ agents in the sequence are set to pick their most favourite bundle of $\lfloor m/n\rfloor$ goods. Subsequently, the remaining agents choose their favourite subsets of items, each containing $\lfloor m/n\rfloor + 1$ goods. Therefore, the agents get to choose a fixed number of items from a feasible set according to a predefined hierarchy. 

\begin{algorithm}[H]
\caption{EFX under Leveled Valuations} \label{LeveledAlg}
\label{alg:LeveledEFX}
\textbf{Input}: An instance with leveled valuations, $m=kn+r$ \\
\textbf{Output}: An EFX allocation
\begin{algorithmic}[1] %[1] enables line numbers
\STATE Select an arbitrary picking order $\sigma = [\sigma_1,\sigma_2,\ldots,\sigma_n]$ 
\STATE Let the first $n-r$ agents (according to $\sigma$) pick their favourite bundle of $\lfloor m/n\rfloor$ available items
\STATE Let the remaining $r$ agents pick their favourite bundle, each consisting of $\lfloor m/n\rfloor + 1$ items 
\STATE \textbf{return} allocation
\end{algorithmic}
\end{algorithm}

\begin{theorem}\label{thm:efx}
    Under general leveled valuations, there exists an EFX allocation. 
\end{theorem}

\begin{proof}
The algorithm is shown as Algorithm \ref{alg:LeveledEFX}. We divide the agents into two levels, namely level $L$ and level $H$; agents in the lower level possess an item less than those in the higher level. Clearly, agents in $H$ attain more value than agents in $L$; they are in fact envy-free towards them. Therefore, the EFX criterion might be violated only among agents residing at different levels. In this case, agents in $L$ might become envious of agents in $H$. Suppose, for the sake of contradiction, that EFX is violated, that is, there is an agent $i$ in $L$ that EFX-envies an agent $j$ in level $H$. This would imply that after removing any item from $A_j$, there exists a bundle of $\lfloor m/n\rfloor$ items that agent $i$ strictly prefers over her own bundle. However, if such a bundle existed, she would have picked it earlier in the allocation process since she preceded $j$ in the picking order. Hence, by contradiction, agent $i$ is EFX towards agent $j$.
\end{proof}

%% file: truthfulness.tex
\section{Truthfulness}

In this section we study the mechanism design aspect of fairly allocating goods among agents. In Section \ref{sec:TruthfulMMS} we show that constant MMS approximations are possible via truthful mechanisms, in contrast to non-leveled valuations. Specifically, we show that any {\em balanced} allocation, which allocates bundles of size $\lfloor m/n \rfloor$, and $\lceil m/n \rceil$,  achieves a  $\frac{k-1}{k}$-MMS approximation under additive leveled valuations for $k \ge 2$, where $k = \lfloor m/n \rfloor$. Furthermore, we show that this bound is tight for balanced allocations. The algorithm described in Section \ref{section:efx} satisfies truthfulness and outputs a balanced allocation, thus, simultaneously guarantees EFX, $\frac{k-1}{k}$-MMS, and truthfulness for additive leveled valuations.\footnote{An earlier version of our work included the claim that the allocation obtained by Algorithm \ref{alg:LeveledEFX} is EFX and $\max \{ 2/3, (k-1)/k\}$-MMS, albeit with an incorrect proof for the case of $2/3$. We thank Mahyar Afshinmehr, Mehrafarin Kazemi, and Kurt Mehlhorn for pointing out a bug in our proof. In fact, they provide an improved bound, showing that the allocation satisfies EFX and $\frac{k}{k+1}$-MMS \citep{afshinmehr2024mmsapproximationsadditiveleveled}.} Then, we show that an $\frac{1}{2}$-MMS approximation can be achieved for subadditive leveled valuations via a truthful mechanism.

In Section \ref{sec:TruthfulEFX} we show a characterization of truthful, EFX, PO, non-bossy, and neutral mechanisms under strict ordinal leveled preferences.

\subsection{Maximin Share Guarantee}\label{sec:TruthfulMMS}

In the case of additive agents, it is known that an EFX allocation is also a $2/3$-MMS allocation for $n \in \{2,3\}$, due to \citet{AmanatidisComparisonsApproximate}. The approximation degrades to $4/7=0.5714$
as the number of agents grows. We proceed to show that balanced allocations achieve a good MMS approximation in general. We call an allocation {\em balanced} if the bundles of any two agents differ in size by at most one.

\begin{proposition}
    Any balanced allocation achieves a $\frac{k-1}{k}$-MMS approximation under additive leveled valuations for $k \ge 2$, where $k = \lfloor m/n \rfloor$. This bound is tight for balanced allocations.
\end{proposition}

\begin{proof}
    Let $B = (B_1, \dots, B_n)$ be a balanced allocation, and let $B' = (B'_1, \dots, B'_n)$ be the MMS allocation of agent $i$. Also, denote the bundle corresponding to the MMS value of agent $i$ by $B_i'$, i.e., $v_i(B'_i) = \mu_i$. Since $B$ is balanced, $|B_i|$ is either equal to $k$ or $k+1$; in the latter case, $i$ trivially guarantees her MMS. We examine the former case. One can observe that due to leveled valuations, $|B'_i|=k$. Also, for any bundle $B_l$ of size $k$ and for any item $j$ of $B'_i$, it holds $v_i(B_l) \ge v_i(B'_i\setminus\{j\})$. By summing up all these inequalities, we get $k \cdot v_i(B_l) \ge \sum_{j \in B'_i} v_i(B_i' \setminus \{j\}) = (k-1) \cdot v_i(B_i')$ which implies that $ v_i(B_l) \geq (k-1)/k \cdot \mu_i$, as needed. 

    Next we show that this bound is tight for balanced allocations. We show that there are instances in which there exists a balanced allocation that guarantees to some agent $i$ only a $\frac{k-1}{k}$ fraction of her MMS. Consider an instance with $n$ agents and $m = kn+1$ items, for some integer parameter $k\geq 1$. We define identical additive valuations with two types of items, namely, a set $L$ of $k$ small items for which each agent has value equal to $\frac{k-1}{k^2} + \epsilon$, for sufficiently small $\epsilon$ and a set $H$ of $k(n-1)+1$ large items for which each agent has value equal to $\frac{1}{k}$. It is easy to verify that the proposed valuations are indeed leveled: for any positive integer $t\leq k$, any $t$ small items are always larger than any $(t - 1)$ large ones.
    The MMS value of each agent is at least $1$: consider the allocation in which $n-1$ agents receive a bundle of $k$ large items, each having a total value of $k \cdot 1/k = 1$, and one agent gets $k$ small items and one large item, yielding a total value of $k \cdot (\frac{k-1}{k^2} +\epsilon) + \frac{1}{k} = 1 + k\epsilon$. Now, consider a balanced allocation that assigns to agent $i$ all the small items; $i$ receives approximately a $\frac{k-1}{k}$-fraction of her MMS for small enough $\epsilon$.
\end{proof}

The mechanism described in Section \ref{section:efx} falls under the broad class of \emph{serial dictatorships with quotas} (see Definition \ref{def:SDQs}). Serial dictatorships with quotas satisfy strategyproofness.

\begin{definition}\label{def:SDQs}[Serial dictatorship with quotas, \citep{PapaiQuotas}] The {\em serial dictatorship with quotas} (SDQ) mechanism is specified by a permutation of the agents $\sigma : N \rightarrow N$ and a quota system $q=(q_1,\dots,q_n)$ such that $\sum_{i=1}^n q_i = m$.     
\end{definition}

\begin{corollary}
    There exists a deterministic truthful mechanism that guarantees EFX and $\frac{k-1}{k}$-MMS for additive leveled valuations, where $k = \lfloor m/n \rfloor$ and $k \ge 2$.
\end{corollary}

Lastly, one can derive analogous propositions for truthful MMS approximations via SDQs for subadditive valuations.

\begin{proposition}
    An $\frac{1}{2}$-MMS approximation under subadditive leveled valuations can be achieved via a truthful mechanism.
\end{proposition}

\begin{proof}
    Consider an SDQ mechanism with the quota structure of Algorithm \ref{alg:LeveledEFX}. We have argued that the bundle that corresponds to the MMS value has cardinality $k$. Let's consider the case of $k \ge 2$; the case of $k < 2$ is treated by Proposition \ref{prop:fewItems} which can be used with an SDQ.  Let's assume that there is an agent $i$ receiving a bundle $B_i$ that does not achieve her maximin share, i.e., $v_i(B_i) \le \mu_i$, otherwise the claim follows. Let  $g\not\in B_i$ be an arbitrary item. Then, we have:
    $$\mu_i \le v_i(B_i \cup g) \le v_i(B_i) + v_i(g) \le 2 \cdot v_i(B_i)$$

\noindent and the claim follows. The inequalities follow from subadditivity and the definition of leveled valuations. 
\end{proof}

\begin{corollary}
    There exists a deterministic truthful mechanism that simultaneously guarantees EFX and $\frac{1}{2}$-MMS for subadditive leveled valuations.
\end{corollary}

Under leveled valuations, we observe
that strategyproofness can be achieved in conjunction with strong fairness guarantees. \citet{AmanatidisBirmpasChristodoulouMarkakis} introduce the notion of \emph{controlling items}; we say that a player controls an item if reporting certain values guarantees her this item. Control is the major hurdle that makes fairness and truthfulness incompatible under additive valuations. We note that this is not applicable to our case, since control of singletons\footnote{Control of singletons implies that whenever an agent strongly desires an item, i.e. prefers it more than all the other items combined, she gets it.} (or pairs) is not possible under leveled valuations.

\subsection{Envy-freeness up to Any Good}\label{sec:TruthfulEFX}

We focus on (strict) ordinal preferences. We give an algorithmic characterization of the mechanisms satisfying a set of desired properties. The same characterization holds for the case of lexicographic valuations \citep{HosseiniLexicographic}.

\textbf{Axiomatic Properties.} We are interested in core axiomatic properties, such as efficiency, non-bossiness, and neutrality. A common notion of economic efficiency is Pareto optimality. Roughly speaking, under a Pareto efficient allocation, no individual can become strictly better off without hurting someone else's value. There exist instances where EFX cannot be satisfied in conjunction with Pareto optimality, even for 2 players \citep{PlautRoughgarden}. Non-bossiness is a property of the allocation rule that ensures that no agent can change (by a deviating bid) the allocation without changing their own bundle. Neutrality states that the allocation rule does not depend on the labeling of the items.

\begin{proposition}\citep{PapaiQuotas}\label{prop:papaiSDQ}
    For leveled preferences, a mechanism is Pareto optimal, strategyproof, non-bossy, and neutral if and only if it is an SDQ.
\end{proposition}

\begin{theorem}
    A mechanism satisfies EFX, PO, truthfulness, non-bossiness, and neutrality if and only if it is a Serial Quota Dictatorship with the quotas of Algorithm \ref{alg:LeveledEFX}.
\end{theorem}

\begin{proof}
    Algorithm \ref{alg:LeveledEFX} is an SDQ with quotas $q_i = \lfloor m/n \rfloor$ for all $i \in [n-r]$ and $q_i = \lfloor m/n \rfloor + 1$ for all $i > n-r$, and therefore EFX according to Theorem \ref{thm:efx}. Moreover, it follows from Proposition \ref{prop:papaiSDQ} that it is PO, truthful, non-bossy, and neutral. We also prove the converse; assume that there exists a mechanism $f$ that satisfies these properties. By Proposition \ref{prop:papaiSDQ}, $f$ must be an SDQ for some $\sigma$ and a specified quota system such that $\sum_{i=1}^n q_i = m$. Assume w.l.o.g. that $m \ge n$. Suppose, for the sake of contradiction, that the quota system of $f$ violates the quota system described by Algorithm \ref{alg:LeveledEFX}; then, we show that EFX is violated. First, note that if there exists a pair of agents $i,j \in N$ such that $\lvert A_i \rvert = \lvert A_j \rvert -2$, then EFX is violated, as agent $i$ will envy agent $j$ due to the definition of leveled valuations, which is a contradiction. So it must be that the sizes of the allocated bundles must differ by at most one. Then we argue that if there is an  agent that receives a smaller size bundle, they must precede those that receive the larger size bundles. Indeed, suppose that agent $j$ picks her bundle before agent $i$, i.e., $\sigma(j) < \sigma(i)$, but $\lvert A_j \rvert = \lvert A_i \rvert + 1$. Then, it may be the case that $j$ has picked $i$'s favorite available bundle of size $\lvert A_i \rvert$ along with some additional item $g$, thus violating EFX after the hypothetical removal of $g$. Therefore, we conclude that $f$ must be an SDQ with the quota system specified in Algorithm \ref{alg:LeveledEFX}.
\end{proof}

\begin{remark}
    Any deterministic strategyproof and non-bossy mechanism is also group-strategyproof \citep{PapaiQuotas}. 
\end{remark}

%% file: conclusion.tex
\section{Conclusion and Discussion}
\label{sec:conclusion}

We explored fairness under leveled valuations and presented a complete set of results regarding this class. We studied the maximin share guarantee and provided strong existence and non-existence results. Furthermore, our analysis improves the current understanding of EFX allocations. We argue that the study of leveled preferences provides meaningful insights for more important classes in the economic paradigm. Our ideas for two agents with submodular valuations might be of independent interest beyond the domain of leveled valuations. Lastly, it is interesting to see whether a strict superset of dictatorial mechanisms achieve better approximation guarantees than the ones presented here. We suspect that picking-exchange mechanisms defined by \citet{AmanatidisBirmpasChristodoulouMarkakis} may do so.

%% file: acknowledgments.tex
\section*{Acknowledgements}

This work has been partially supported by project MIS 5154714 of the National Recovery and Resilience Plan Greece 2.0 funded
by the European Union under the NextGenerationEU Program.